\documentclass[a4paper,fleqn,11pt]{article}

\usepackage{amsmath}
\usepackage{amsthm}
\usepackage{amssymb}
\usepackage{accents}

\usepackage[a4paper,top=3cm, bottom=3cm, left=3cm, right=3cm]{geometry}
\usepackage[shortlabels,inline]{enumitem}
\usepackage[pdftex]{graphicx}
\usepackage{subcaption}
\usepackage[pdftex,ocgcolorlinks,pagebackref=false]{hyperref}
\usepackage[affil-it]{authblk}

\setlist[enumerate,1]{label={(\roman*)}}

\theoremstyle{plain}
\newtheorem{theorem}{Theorem}[section]
\newtheorem{lemma}[theorem]{Lemma}
\newtheorem{proposition}[theorem]{Proposition}
\newtheorem{remark}[theorem]{Remark}

\theoremstyle{definition}
\newtheorem{definition}[theorem]{Definition}

\DeclareMathOperator{\ev}{ev}
\newcommand{\setbuild}[2]{\left\{#1\middle|#2\right\}}
\newcommand{\positiveintegers}{\mathbb{N}_{>0}}
\newcommand{\naturals}{\mathbb{N}}
\newcommand{\reals}{\mathbb{R}}
\newcommand{\complexes}{\mathbb{C}}
\newcommand{\nonnegativereals}{\mathbb{R}_{\ge 0}}

\newcommand{\distributions}[1][]{\mathcal{P}_{#1}}
\newcommand{\norm}[2][]{\left\|#2\right\|_{#1}}
\DeclareMathOperator{\GHZ}{GHZ}
\DeclareMathOperator{\EPR}{EPR}
\DeclareMathOperator{\Tr}{Tr}
\newcommand{\entropy}{H}

\newcommand{\ket}[1]{\left|#1\right\rangle}
\newcommand{\bra}[1]{\left\langle #1\right|}
\newcommand{\ketbra}[2]{\left|#1\middle\rangle\!\middle\langle#2\right|}
\newcommand{\braket}[2]{\left\langle#1\middle|#2\right\rangle}
\newcommand{\fidelity}{F}
\newcommand{\purifieddistance}{D}
\DeclareMathOperator{\states}{\mathcal{S}}
\newcommand{\loccto}[1][]{\xrightarrow{\textnormal{LOCC}}_{#1}}
\newcommand{\loccqrate}[2]{R_{\textnormal{LOCCq}}(#1\to #2)}
\newcommand{\loccrate}[2]{R_{\textnormal{LOCC}}(#1\to #2)}
\newcommand{\functionals}[1]{\mathcal{F}_{#1}}
\newcommand{\purefunctionals}[1]{\mathcal{F}^{\textnormal{pure}}_{#1}}

\title{Asymptotic continuity of additive entanglement measures}
\author[1,2]{P\'eter Vrana}
\affil[1]{Institute of Mathematics, Budapest University of Technology and Economics, Egry~J\'ozsef u.~1., Budapest, 1111 Hungary.}
\affil[2]{MTA-BME Lend\"ulet Quantum Information Theory Research Group, Budapest, Hungary}

\begin{document}
\maketitle

\begin{abstract}
We study rates asymptotic of transformations between entangled states by local operations and classical communication and a sublinear amount of quantum communication. It is known that additive asymptotically continuous entanglement measures provide upper bounds on the rates that are achievable with asymptotically vanishing error. We show that for transformations between pure states, the optimal rate between any pair of states can be characterized as the infimum of such upper bounds provided by fully additive asymptotically continuous entanglement measures.
\end{abstract}

\section{Introduction}

The uniqueness theorem singles out the entropy of entanglement as the essentially unique entanglement measure for pure bipartite states, in the context of asymptotic transformations by local operations and classical communication (LOCC) \cite{popescu1997thermodynamics}. A closer examination of the assumptions reveals that even though there exist other quantities that do not increase under LOCC, the only one that survives in the asymptotic limit is the entropy of entanglement \cite{vidal2000entanglement}. The main message of the uniqueness result is that, asymptotically, there is a single kind of pure bipartite entanglement and the states only differ in the amount of entanglement they contain, as measured by a single number.

The situation changes when mixed or multipartite states are considered. The generalization of the uniqueness theorem to mixed bipartite states has a weaker conclusion: any additive (i.e. $E(\rho^{\otimes n})=nE(\rho)$), asymptotically continuous and normalized entanglement measure lies between the distillable entanglement $E_D$ \cite{bennett1996mixed,rains1999rigorous} and the entanglement cost $E_C$ \cite{hayden2001asymptotic}, two operationally defined entanglement measures quantifying the amount of pure entanglement (in the form of Bell pairs) into or from which the state in question can asymptotically be transformed (see \cite{horodecki2000limits,donald2002uniqueness} for the precise assumptions).

More generally, one can consider rates of transformations between any pair of states \cite{bennett1996mixed,horodecki2003rates}. If $\loccrate{\rho}{\sigma}$ denotes the maximum rate for asymptotic transformations of $\rho$ into $\sigma$, and $E$ is an asymptotically continuous additive LOCC-monotone, then the inequality
\begin{equation}
\loccrate{\rho}{\sigma}\le\frac{E(\rho)}{E(\sigma)}
\end{equation}
holds. However, it is not immediately clear how tight such an upper bound can be. In fact, even the existence of such a functional $E$ is not trivial, and the only known non-operationally defined quantities with these properties are the regularized relative entropy of entanglement \cite{vedral1997quantifying,christandl2006structure}, the squashed entanglement \cite{christandl2004squashed,alicki2004continuity} and its multipartite variants \cite{yang2009squashed}, and the conditional entanglement of mutual information \cite{yang2008additive}.

Another question is to what extent are the conditions on $E$ necessary. While additivity is a reasonable requirement in an asymptotic setting (and can be enforced by considering the regularization instead), asymptotic continuity is apparently an \emph{ad hoc} condition (although suggested by the Fannes inequality \cite{fannes1973continuity}, which implies a continuity estimate for the entropy of entanglement), and indeed different versions thereof have been considered in the literature. It should also be noted that the logarithmic negativity provides an upper bound on the distillable entanglement even though it is \emph{not} asymptotically continuous \cite{vidal2002computable}.

Our main result can be viewed as a partial answer to these questions. We will consider a variant of the transformation rates defined in \cite{bennett2000exact} and further investigated in \cite{thapliyal2003multipartite}. Let $\loccqrate{\rho}{\sigma}$ be the maximum rate at which $\rho$ can be transformed into $\sigma$ by LOCC transformations and a sublinear amount of quantum communication, with asymptotically vanishing error. We prove the following characterization of the rates for transformations between pure states in terms of fully additive (i.e. $E(\rho\otimes\sigma)=E(\rho)+E(\sigma)$), asymptotically continuous (in the sense of Definition~\ref{def:asymptoticcontinuity} below) entanglement measures:
\begin{theorem}\label{thm:main}
For every pair of pure $k$-partite states $\ket{\varphi}$ and $\ket{\psi}$ the largest achievable rate is
\begin{equation}\label{eq:ratecharacterization}
\loccqrate{\ketbra{\varphi}{\varphi}}{\ketbra{\psi}{\psi}}=\inf_{\substack{E  \\  E(\ket{\psi})\neq 0}}\frac{E(\ket{\varphi})}{E(\ket{\psi})},
\end{equation}
where the infimum is over functions $E$ on $k$-partite pure states of arbitrary dimension that are
\begin{itemize}
\item normalized to $1$ on the Greenberger--Horne--Zeilinger (GHZ) state,
\item fully additive,
\item monotone on averege under LOCC,
\item asymptotically continuous.
\end{itemize}
\end{theorem}
In addition, we show that any fully additive functional on pure states that is nonincreasing under asymptotic LOCC transformations is necessarily asymptotically continuous, with an explicit continuity estimate. The picture for mixed states is less clear at the moment. While we can show that it is still sufficient to consider fully additive measures, our proof method does not seem to be sufficiently powerful to show that one can restrict to asymptotically continuous ones and still obtain a characterization like \eqref{eq:ratecharacterization}.

In order to efficiently present our partial results for mixed states and the stronger result for pure states, the proof of our main result is split into two parts. In Section~\ref{sec:general} we characterize $\loccqrate{\rho}{\sigma}$ for arbitrary (pure or mixed) states in terms of fully additive functions that are monotone under asymptotic LOCC transformations. The main tool here is a reduction to general results in the mathematical theory of resources. The second part is the content of Section~\ref{sec:pure} where, now restricting to pure states, we find equivalent conditions for monotonicity under asymptotic LOCC transformations, assuming additivity and monotonicity on average under LOCC. Among these are asymptotic continuity (which is clearly sufficient). Interestingly, one of the equivalent conditions is purely algebraic and reminiscent of the chain rule satisfied by the Shannon entropy:
\begin{equation}
E(\sqrt{p}\varphi\oplus\sqrt{1-p}\psi)=pE(\varphi)+(1-p)E(\psi)+h(p),
\end{equation}
where $p\in[0,1]$ and $\oplus$ is the direct sum, i.e. superposition of locally orthogonal vectors, and $h(p)$ is the binary entropy. It is worth pointing out the connection to the uniqueness theorem as an illustration: for bipartite pure states (more generally, generalized GHZ states), this property together with vanishing on separable states suffices to ensure that $E$ is equal to the Shannon entropy of the Schmidt coefficients.

\section{Notations}\label{sec:notations}

All logarithms are to base $2$. The binary entropy function is $h(p)=-p\log p-(1-p)\log(1-p)$.

The number of subsystems $k$ will be fixed throughout. $\mathcal{H},\mathcal{K}$ refer to Hilbert spaces of composite systems with e.g. $\mathcal{H}=\mathcal{H}_1\otimes\cdots\otimes\mathcal{H}_k$. We denote the set of states on $\mathcal{H}$ by $\states(\mathcal{H})$. For $\rho\in\states(\mathcal{H})$ and $\sigma\in\states(\mathcal{K})$ the tensor product $\rho\otimes\sigma$ is regarded as a $k$-partite state with the grouping $(\mathcal{H}_1\otimes\mathcal{K}_1)\otimes\cdots\otimes(\mathcal{H}_k\otimes\mathcal{K}_k)=:\mathcal{H}\otimes\mathcal{K}$.

It will be useful to consider a kind of sum operation on $k$-partite Hilbert spaces, based on the direct sums of the local Hilbert spaces:
\begin{equation}
(\mathcal{H}_1\oplus\mathcal{K}_1)\otimes\cdots\otimes(\mathcal{H}_k\oplus\mathcal{K}_k).
\end{equation}
Note that this product contains both $\mathcal{H}$ and $\mathcal{K}$ as subspaces, orthogonal to each other. For vectors $\varphi\in\mathcal{H}$ and $\psi\in\mathcal{K}$ we consider the direct sum $\varphi\oplus\psi$ as an element of this product space. When $\varphi,\psi$ are vectors in the same Hilbert space $\mathcal{H}$, their direct sum can also be viewed as $\varphi\otimes\ket{00\ldots 0}+\psi\otimes\ket{11\ldots 1}\in\mathcal{H}\otimes(\complexes^2)^{\otimes k}$, up to a local unitary transformation.

For example, the generalized GHZ state is the pure state with state vector
\begin{equation}
\frac{1}{\sqrt{r}}(\ket{11\ldots 1}+\ket{22\ldots 2}+\cdots+\ket{rr\ldots r}),
\end{equation}
which is local unitary equivalent to the direct sum of $r$ copies of $\frac{1}{\sqrt{r}}\ket{00\ldots 0}$. We will denote the corresponding state by $\GHZ_r$, omitting the subscript when $r=2$. By a slight abuse of notation, we will write $\varphi,\psi,\GHZ_r,\ldots$ both for the unit vectors and the state determined by them. 

We equip each state space $\states(\mathcal{H})$ with the purified distance $\purifieddistance(\rho,\sigma)=\sqrt{1-\fidelity(\rho,\sigma)^2}$ \cite[Definition 4.]{tomamichel2010duality} (see also \cite{gilchrist2005distance}), where
\begin{equation}\label{eq:fidelity}
\fidelity(\rho,\sigma)=\Tr\sqrt{\sigma^{1/2}\rho\sigma^{1/2}}
\end{equation}
is the fidelity. The purified distance is a metric that in addition satisfies
\begin{equation}
\purifieddistance(\rho_1\otimes\rho_2,\sigma_1\otimes\sigma_2)\le \purifieddistance(\rho_1,\sigma_1)+\purifieddistance(\rho_2,\sigma_2).
\end{equation}
Completely positive trace-preserving maps are contractive with respect to the purified distance.

We will write $\rho\loccto\sigma$ if there is a channel $\Lambda$ that can be implemented via local operations and classical communication (an LOCC channel) and $\Lambda(\rho)=\sigma$. Similar notation will be used for approximate transformations: $\rho\loccto[\epsilon]\sigma$ means that there exists another state $\sigma'$ (on the same space as $\sigma$) such that $\rho\loccto\sigma'$ and $\purifieddistance(\sigma',\sigma)\le\epsilon$ (in particular, $\loccto$ is the same as $\loccto[0]$). Using the properties of the purified distance and that compositions and tensor products of LOCC channels are LOCC channels, one can see the implications
\begin{equation}
(\rho\loccto[\epsilon_1]\sigma\text{ and }\sigma\loccto[\epsilon_2]\tau)\implies(\rho\loccto[\epsilon_1+\epsilon_2]\tau)
\end{equation}
and
\begin{equation}
(\rho_1\loccto[\epsilon_1]\sigma_1\text{ and }\rho_2\loccto[\epsilon_2]\sigma_2)\implies(\rho_1\otimes\rho_2\loccto[\epsilon_1+\epsilon_2]\sigma_1\otimes\sigma_2).
\end{equation}
In addition, LOCC transformations between pure states have the following compatibility with the direct sum \cite[Proposition 2.]{jensen2019asymptotic}:
\begin{equation}
(\varphi_1\loccto\psi_1\text{ and }\varphi_2\loccto\psi_2)\implies\sqrt{p}\varphi_1\oplus\sqrt{1-p}\varphi_2\loccto\sqrt{p}\psi_1\oplus\sqrt{1-p}\psi_2.
\end{equation}

\section{Asymptotic LOCC transformations}\label{sec:general}

In the setting of Shannon theory, we consider LOCC transformations in the asymptotic limit of many copies, allowing approximate transformations with an error approaching $0$, and assisted with a sublinear number of additional GHZ states. This is called LOCCq in \cite{bennett2000exact}, where the ``q'' stands for quantum communication of $o(n)$ qubits for $n$ copies, which is equivalent to $o(n)$ GHZ states.
\begin{definition}\label{def:achievablerate}
Let $\rho,\sigma$ be $k$-partite states (possibly on different Hilbert spaces, which we leave implicit). A number $r\in\nonnegativereals$ is an \emph{achievable rate} (for transforming $\rho$ into $\sigma$) if
\begin{equation}
\forall\delta>0:\limsup_{n\to\infty}\inf\setbuild{\epsilon\in\nonnegativereals}{\rho^{\otimes n}\otimes\GHZ^{\otimes\lfloor\delta n\rfloor}\loccto[\epsilon]\sigma^{\otimes\lceil rn\rceil}}=0.
\end{equation}
The supremum of achievable rates will be denoted by $\loccqrate{\rho}{\sigma}$.
\end{definition}
If $\rho$ is distillable, i.e. for every $\epsilon$ there is an $n$ such that $\rho^{\otimes n}\loccto[\epsilon]\GHZ$, then we obtain the same supremum if we require transformations without the sublinear supply of GHZ states \cite{thapliyal2003multipartite} (since these can be obtained from a small number of copies of $\rho$ without changing the rate \cite{maneva2002improved,chen2007multi}). In particular, the entanglement cost of a bipartite state $\rho$ satisfies $E_C(\rho)=\loccrate{\EPR}{\rho}^{-1}=\loccqrate{\EPR}{\rho}^{-1}$, and if $E_D(\rho)=\loccrate{\rho}{\EPR}>0$, then also $E_D(\rho)=\loccqrate{\rho}{\EPR}$.

When $k=2$ and both $\rho$ and $\sigma$ are pure states, the values of $\loccqrate{\rho}{\sigma}$ and of $\loccrate{\rho}{\sigma}$ are equal to $\frac{\entropy(\Tr_1\rho)}{\entropy(\Tr_1\sigma)}$ \cite{bennett1996concentrating,thapliyal2003multipartite}. For mixed states or when $k\ge 3$, the problem of determining either $\loccqrate{\rho}{\sigma}$ or $\loccrate{\rho}{\sigma}$ in general is wide open.

This problem fits in the general framework of resource theories, in particular in the mathematical framework of preordered commutative monoids, as we will see below. We briefly recall the required definitions and results. We use a multiplicative notation $1,xy,x^n,\ldots$, which aligns better with tensor products and powers, but otherwise follow \cite{fritz2017resource}. A preordered commutative monoid gives rise to an ordered commutative monoid by identifying $x$ and $y$ whenever both $x\ge y$ and $x\le y$ hold. The results of \cite{fritz2017resource} can be applied after this identification and translated back to the preordered setting when convenient.
\begin{definition}
A \emph{preordered commutative monoid} is a set $M$ equipped with a binary operation $\cdot$ that is associative and commutative, and has a neutral element $1$; a preorder $\le$ (i.e. a reflexive and transitive relation); such that $x\ge y$ implies $xz\ge yz$ for all $x,y,z\in M$.
\end{definition}
In our case $M$ will be the set of (equivalence classes of) $k$-partite states, the operation is the tensor product, and the preorder is given by asymptotic LOCC transformations (see below for details).

Our setting is special in that the inequality $x\ge 1$ holds for all $x\in M$. From now on we will assume this property. Compared to the general situation treated in \cite{fritz2017resource}, this results in simplifications of some of the definitions and formulas. Here we state only the special forms that take advantage of this fact (see \cite[3.19. Remark]{fritz2017resource}).
\begin{definition}
An element $g\in M$ is a \emph{generator} if for every $x\in M$ there exists an $n\in\naturals$ such that $g^n\ge x$.
\end{definition}

\begin{definition}
A \emph{functional} on the preordered commutative monoid $M$ is a map $f:M\to\reals$ satisfying $f(xy)=f(x)+f(y)$ and $x\ge y\implies f(x)\ge f(y)$ for all $x,y\in M$.

Let $g$ be a generator and $x,y\in M$. $r\in\nonnegativereals$ is a \emph{regularized rate} from $x$ to $y$ if for every $\delta>0$ and neighbourhood $U$ of $r$ there is a fraction $\frac{m}{n}\in U$ and $d\in\naturals$ such that $d\le\delta\max(m,n)$ and $x^ng^d\ge y^m$.
\end{definition}
Functionals provide upper bounds on regularized rates since $x^ng^d\ge y^m$ implies $f(x)+\frac{d}{n}f(g)\ge\frac{m}{n}f(y)$, therefore $f(x)\ge rf(y)$ for every regularized rate $r$. A central result is that regularized rates are in fact characterized by functionals:
\begin{theorem}[{\cite[8.24. Theorem]{fritz2017resource}}]\label{thm:regularizedratefromfunctionals}
The supremum of regularized rates from $x$ to $y$ is equal to
\begin{equation}\label{eq:regularizedrateformula}
\inf_f\frac{f(x)}{f(y)},
\end{equation}
where the infimum ranges over functionals $f$ that satisfy $f(y)\neq 0$.
\end{theorem}

\begin{remark}\label{rem:normalization}
Nonnegative multiples of functionals are also functionals, and the ratios in the regularized rate formula are not sensitive to such rescaling. The only functional that evaluates to $0$ on the generator is the zero functional, since $1\le x\le g^n$ implies $0=f(1)\le f(x)\le nf(g)$. Therefore in \eqref{eq:regularizedrateformula} we may restrict to functionals satisfying $f(g)=1$.
Normalized functionals form a convex set that is also compact with respect to the weak-* topology (the smallest topology that makes every evaluation map $\ev_m:f\mapsto f(m)$ continuous, where $m\in M$).
\end{remark}

We will make use of an alternative way of viewing regularized rates, provided by the following lemma:
\begin{lemma}\label{lem:alternativeregularizedrate}
Let $x,y\in M$. The following are equivalent:
\begin{enumerate}
\item\label{it:originalregularizedrate} $r$ is a regularized rate from $x$ to $y$,
\item\label{it:alternativeregularizedrate} $\forall\delta>0\exists n\in\positiveintegers: x^ng^{\lfloor\delta n\rfloor}\ge y^{\lceil rn\rceil}$.
\end{enumerate}
\end{lemma}
\begin{proof}
Let $c\in\naturals$ such that $y\le g^c$.

\ref{it:originalregularizedrate}$\implies$\ref{it:alternativeregularizedrate}: Suppose that $r$ is a regularized rate and let $\delta>0$. With $\delta'=\frac{\delta}{3(r+1)}$ and the open set $U=(r-\frac{\delta}{3c},r+1)$ choose $m,n\in\positiveintegers$ such that $\frac{m}{n}\in U$ and $x^ng^{\lfloor\delta'\max(m,n)\rfloor}\ge y^m$ (possible by the definition of a regularized rate and using $g\ge 1$). For every $t\in\naturals$ we also have $x^{tn}g^{\lfloor \delta'\max(tm,tn)\rfloor}\ge x^{tn}g^{t\lfloor\delta'\max(m,n)\rfloor}\ge y^{tm}$, therefore we may assume that $n>\frac{3c}{\delta}$ (multiplying $m$ and $n$ with a large natural number, if necessary).

If $m\ge\lceil rn\rceil$, then
\begin{equation}
y^{\lceil rn\rceil}\le y^m\le x^ng^{\lfloor \delta'\max(m,n)\rfloor},
\end{equation}
where the exponent of $g$ is upper bounded by $\delta'(r+1)n\le\delta n/3\le\delta n$.

Otherwise the choices ensure that
\begin{equation}
y^{\lceil rn\rceil} = y^my^{\lceil rn\rceil-m}\le x^ng^{\lfloor\delta'\max(m,n)\rfloor+c(\lceil rn\rceil-m)}
\end{equation}
and the exponent of $g$ satisfies
\begin{equation}
\begin{split}
\frac{1}{n}\left[\lfloor\delta'\max(m,n)\rfloor+c(\lceil rn\rceil-m)\right]
 & \le \frac{1}{n}\left[\delta'(r+1)n+c(1+rn-(r-\frac{\delta}{3c})n)\right]  \\
 & =  \delta'(r+1)+\frac{c}{n}+c\frac{\delta}{3c}\le\delta.
\end{split}
\end{equation}
Since the exponent is an integer upper bounded by $\delta n$, it is at most $\lfloor\delta n\rfloor$. By $g\ge 1$ we can replace the exponent with the upper bound $\lfloor\delta n\rfloor$ to get $x^ng^{\lfloor\delta n\rfloor}\ge y^{\lceil rn\rceil}$.

\ref{it:alternativeregularizedrate}$\implies$\ref{it:originalregularizedrate}: Suppose that \ref{it:alternativeregularizedrate} holds and let $\delta>0$ and $U$ a neighbourhood of $r$. Choose $n\ge 1$ such that $x^ng^{\lfloor\delta n\rfloor }\ge y^{\lceil rn\rceil}$. Then for every $t\in\naturals$ the inequality
\begin{equation}
x^{tn}g^{\lfloor\delta tn\rfloor}\ge x^{tn}g^{t\lfloor\delta n\rfloor}\ge y^{t\lceil rn\rceil}\ge y^{\lceil rtn\rceil}
\end{equation}
also holds. Therefore we can choose $n$ so large that $\frac{\lceil rn\rceil}{n}\in U$. Since $\lfloor\delta n\rfloor\le\lfloor\delta n\max(n,\lceil rn\rceil)\rfloor$, we conclude that $r$ is a regularized rate.
\end{proof}

Next we begin the construction of our preordered commutative monoid by defining an equivalence relation on $k$-partite states. Let $\rho\in\states(\mathcal{H})$ and $\sigma\in\states(\mathcal{K})$. We say that these states are equivalent and write $\rho\sim\sigma$ if there are unitaries $U_j\in U(\mathcal{H}_j\oplus\mathcal{K}_j)$ for all $j=1,\ldots,k$ such that
\begin{equation}
(U_1\otimes\cdots\otimes U_k)(\rho\oplus 0)(U_1\otimes\cdots\otimes U_k)^*=0\oplus\sigma,
\end{equation}
where $0$ on the left (right) hand side is the zero operator on $\mathcal{K}$ ($\mathcal{H}$), and $\rho\oplus 0$ and $\sigma\oplus 0$ are regarded as operators on $(\mathcal{H}_1\oplus\mathcal{K}_1)\otimes\cdots\otimes(\mathcal{H}_k\oplus\mathcal{K}_k)$, supported on the subspace $\mathcal{H}\oplus\mathcal{K}$. In simpler terms, two states are equivalent if they are the same up enlarging the local Hilbert spaces and to unitary equivalence.

It is clear that every state is equivalent to some state on $\complexes^d\otimes\cdots\otimes\complexes^d$ when the local dimension $d$ is sufficiently large. To avoid set-theoretical issues, we therefore define $M$ to be
\begin{equation}
M=\left(\bigcup_{d=1}^\infty\states(\complexes^d\otimes\cdots\otimes\complexes^d)\right)/\sim.
\end{equation}
The tensor product of states descends to a well-defined operation on $M$, which is associative, commutative, and has a unit $1$, the equivalence class of separable pure states. This operation turns $M$ into a commutative monoid.

While working with such equivalence classes is necessary for obtaining the required algebraic structure, we do not wish to carry the notational burden that comes with distinguishing a state from its equivalence class. In addition, it offers more flexibility to consider states on finite dimensional Hilbert spaces that are not of the form $\complexes^d\otimes\cdots\otimes\complexes^d$, and it is safe to do so as long as every definition respects the relation of equivalence. For this reason, we will regard states $\rho\in\states(\mathcal{H})$ as elements of $M$, keeping in mind that each state corresponds to a unique equivalence class. In the same spirit, we will use the notation $\otimes$ for the operation on $M$.

The next ingredient that we need is a preorder on $M$ which is compatible with the tensor product. We make the following definition.
\begin{definition}
Let $\rho$ and $\sigma$ be $k$-partite states. We declare $\rho\ge\sigma$ if
\begin{equation}
\limsup_{n\to\infty}\inf\setbuild{\epsilon\in\nonnegativereals}{\rho^{\otimes n}\loccto[\epsilon]\sigma^{\otimes n}}=0.
\end{equation}
\end{definition}
It is straightforward to verify that $\ge$ is well-defined on $M$, and gives a reflexive and transitive relation that is compatible with the multiplication. Clearly $\rho\loccto\sigma$ implies $\rho\ge\sigma$.

As a generator we may choose the GHZ state. It is indeed a generator: a GHZ state can be transformed into an EPR pair between any pair of the parties, which can then be used to teleport any state locally prepared by one party, given sufficient supply of the GHZ states. Therefore for every $\rho$ and sufficiently large $n$ we have $\GHZ^{\otimes n}\loccto\rho$, i.e. $\GHZ^{\otimes n}\ge\rho$. A more careful analysis of this idea shows $\GHZ_{\dim\mathcal{H}}\ge\rho$ if $\rho\in\states(\mathcal{H})$ (or even $\dim\mathcal{H}-\max_j\dim\mathcal{H}_j$).

The following definition gives a name to the set of normalized functionals on $M$ (by Remark~\ref{rem:normalization} these are the only ones that we need to consider).
\begin{definition}\label{def:normalizedfunctionals}
$\functionals{k}$ is the set of maps $E$ from $k$-partite states to $\reals$ which satisfy for all $\rho,\sigma$
\begin{enumerate}
\item $E(\GHZ)=1$,
\item $E(\rho\otimes\sigma)=E(\rho)+E(\sigma)$,
\item if $\limsup_{n\to\infty}\inf\setbuild{\epsilon\in\nonnegativereals}{\rho^{\otimes n}\loccto[\epsilon]\sigma^{\otimes n}}$ then $E(\rho)\ge E(\sigma)$.
\end{enumerate}
\end{definition}
A basic consequence of these properties and the relation $\GHZ_{\dim\mathcal{H}}\ge\rho$ for a state $\rho\in\states(\mathcal{H})$ is that any element $E\in\functionals{k}$ satisfies $0\le E(\rho)\le\log\dim\mathcal{H}$. Known elements of $\functionals{k}$ include the squashed entanglement and its multipartite generalizations \cite{christandl2004squashed,alicki2004continuity,yang2009squashed} and the conditional entanglement of mutual information \cite{yang2008additive}. The relative entropy of entanglement is not in $\functionals{k}$, because it is not additive \cite{vollbrecht2001entanglement}, but it is asymptotically continuous \cite{donald1999continuity}. On the other hand, the regularized relative entropy of entanglement is additive (by definition) and asymptotically continuous \cite[Proposition 3.23]{christandl2006structure}, and it is an open question if it is fully additive, a property which would make it an element of $\functionals{k}$.
\begin{remark}
The defining properties of the normalized functionals are among the strongest axioms considered in the theory of entanglement measures. In particular, the are known to imply convexity (see e.g. \cite[Proposition 3.10]{christandl2006structure}), monotonicity on average, i.e.
\begin{equation}
\left(\rho\loccto\sum_{x\in\mathcal{X}}P(x)\ketbra{x}{x}\otimes\sigma_x\right)\implies E(\rho)\ge\sum_{x\in\mathcal{X}}P(x)E(\sigma_x)
\end{equation}
where $\ketbra{x}{x}$ is a classical ``flag'' state available to all parties (equivalently: one party), and the condition
\begin{equation}
E\left(\sum_{x\in\mathcal{X}}P(x)\ketbra{x}{x}\otimes\sigma_x\right)=\sum_{x\in\mathcal{X}}P(x)E(\sigma_x).
\end{equation}
\end{remark}

At this point we can conclude that it is possible to define regularized rates on $M$, and their supremum is characterized by Theorem~\ref{thm:regularizedratefromfunctionals} in terms of normalized functionals. To connect to the problem set out at the beginning of this section, we show that regularized rates on $M$ and achievable rates in the sense of Definition~\ref{def:achievablerate} are the same.
\begin{proposition}
Let $\rho,\sigma$ be $k$-parite states and $r\in\nonnegativereals$. The following are equivalent:
\begin{enumerate}
\item\label{it:regularized} $r$ is a regularized rate from $\rho$ to $\sigma$,
\item\label{it:achievable} $r$ is an achievable rate, i.e. $\loccqrate{\rho}{\sigma}\ge r$.
\end{enumerate}
\end{proposition}
\begin{proof}
Let $c\in\naturals$ such that $\GHZ^{\otimes c}\loccto\sigma$.

\ref{it:regularized}$\implies$\ref{it:achievable}: Suppose that $r$ is a regularized rate and let $\delta>0$. Let $\delta'=\delta/2$. By Lemma~\ref{lem:alternativeregularizedrate}, there is an $n\ge 1$ such that $\rho^{\otimes n}\otimes\GHZ^{\otimes\lfloor\delta' n\rfloor}\ge\sigma^{\lceil rn\rceil}$. In detail,
\begin{equation}
\limsup_{t\to\infty}\inf\setbuild{\epsilon\in\nonnegativereals}{\rho^{\otimes tn}\otimes\GHZ^{\otimes t\lfloor\delta' n\rfloor}\loccto[\epsilon]\sigma^{t\otimes\lceil rn\rceil}}=0.
\end{equation}
For $N\in\naturals$ let $t=\lfloor\frac{N}{n}\rfloor$. Then $t\to\infty$ as $N\to\infty$, therefore for any $\epsilon>0$ and sufficiently large $N$ we have
\begin{equation}\label{eq:rhodeltaprimeapproximatelytosigma}
\rho^{\otimes tn}\otimes\GHZ^{\otimes t\lfloor\delta' n\rfloor}\loccto[\epsilon]\sigma^{\otimes t\lceil rn\rceil}.
\end{equation}

If $t\lceil rn\rceil\ge\lceil rN\rceil$ then also $\rho^{\otimes tn}\otimes\GHZ^{\otimes t\lfloor\delta' n\rfloor}\loccto[\epsilon]\sigma^{\lceil rN\rceil}$, and here the number of GHZ states is at most $t\delta' n=t\delta n/2\le N\delta$. Otherwise we combine \eqref{eq:rhodeltaprimeapproximatelytosigma} with
\begin{equation}
\GHZ^{\otimes c(\lceil rN\rceil-t\lceil rn\rceil)}\loccto\sigma^{\otimes\lceil rN\rceil-t\lceil rn\rceil},
\end{equation}
and get (using $N\ge tn$)
\begin{equation}
\rho^{\otimes N}\otimes\GHZ^{\otimes t\lfloor\delta' n\rfloor+c(\lceil rN\rceil-t\lceil rn\rceil)}\loccto[\epsilon]\sigma^{\otimes \lceil rN\rceil}.
\end{equation}
The required number of GHZ states satisfies
\begin{equation}
\begin{split}
t\lfloor\delta' n\rfloor+c(\lceil rN\rceil-t\lceil rn\rceil)
 & = \left\lfloor\frac{N}{n}\right\rfloor\lfloor\delta' n\rfloor+c\left(\lceil rN\rceil-\left\lfloor\frac{N}{n}\right\rfloor\lceil rn\rceil\right)  \\
 & \le \frac{N}{n}\delta' n+c\left(1+rN-\left(\frac{N}{n}-1\right) rn\right)  \\
 & = \delta' N+c(1+rn)\le\delta N
\end{split}
\end{equation}
if $N\ge c(1+rn)/\delta'$. As $\delta$ and $\epsilon$ can be arbitrarily small, $r$ is an achievable rate.

\ref{it:achievable}$\implies$\ref{it:regularized}: Let $r$ be an achievable rate and let $\delta>0$. Choose $\delta'=\delta/3$ and let $n=\lceil 3c/\delta\rceil$. Since $r$ is achievable, for all $\epsilon>0$ and sufficiently large $n'$ the relation
\begin{equation}
\rho^{\otimes n'}\otimes\GHZ^{\otimes \lfloor\delta' n'\rfloor}\loccto[\epsilon]\sigma^{\otimes\lceil rn'\rceil}
\end{equation}
holds. Specializing to $n'=tn$ with large $t$ and using $\lceil rtn\rceil\le t\lceil rn\rceil\le\lceil rtn\rceil+t$, we have
\begin{equation}
\rho^{\otimes tn}\otimes\GHZ^{\otimes \lfloor\delta'tn\rfloor+c(t\lceil rn\rceil-\lceil rtn\rceil)}\loccto[\epsilon]\sigma^{\otimes\lceil rtn\rceil}\otimes\sigma^{\otimes(t\lceil rn\rceil-\lceil rtn\rceil)}=\sigma^{\otimes t\lceil rn\rceil},
\end{equation}
where the number of GHZ states satisfies
\begin{equation}
\lfloor\delta'tn\rfloor+c(t\lceil rn\rceil-\lceil rtn\rceil)
 \le t(\delta'n+c) \le t\lfloor\delta n\rfloor.
\end{equation}
This means that $\rho^{\otimes n}\otimes\GHZ^{\otimes\lfloor\delta n\rfloor}\ge\sigma^{\otimes \lceil rn\rceil}$. Since $\delta$ was arbitrarily small, $r$ is a regularized rate by Lemma~\ref{lem:alternativeregularizedrate}.
\end{proof}

To summarize, the results above identify the regularized rates for $M$ as achieveble rates for asymptotic LOCC transformations assisted by a sublinear number of GHZ states (equivalently: sublinear qubits of quantum communication), while Theorem~\ref{thm:regularizedratefromfunctionals} characterizes them in terms of functionals. Thus we have proved the following theorem:
\begin{theorem}\label{thm:loccqrateinf}
For all $\rho,\sigma$ we have
\begin{equation}
\loccqrate{\rho}{\sigma}=\inf_{\substack{E\in\functionals{k}  \\  E(\sigma)\neq 0}}\frac{E(\rho)}{E(\sigma)}.
\end{equation}
\end{theorem}

As a simple application, we obtain the following characterization of the entanglement cost:
\begin{equation}\label{eq:ecmax}
E_C(\rho)=\frac{1}{\loccqrate{\EPR}{\rho}}=\max_{E\in\functionals{2}}E(\rho),
\end{equation}
where writing maximum is justified by compactness of $\functionals{2}$ and continuity of $E\mapsto E(\rho)$. Similarly, when $E_D(\rho)>0$, we have
\begin{equation}\label{eq:edmin}
E_D(\rho)=\loccqrate{\rho}{\EPR}=\min_{E\in\functionals{2}}E(\rho).
\end{equation}
These statements can be seen as a strong duality type extension of the uniqueness theorem \cite{horodecki2000limits,donald2002uniqueness}. The extension is not trivial, as neither $E_C$ nor $E_D$ are known to be an element of $\functionals{2}$. It is not known whether $E_C$ is fully additive \cite{brandao2007remarks} (although it is convex), and there is some evidence that $E_D$ is not fully additive and not convex \cite{shor2001nonadditivity}. We note that \eqref{eq:ecmax} expresses $E_C$ as a maximum of convex functions, which is also convex, while \eqref{eq:edmin} expresses $E_D$ (at least when it is not zero) as a minimum of convex functions, which is in general not convex.

For more than two parties, similar formulas hold with the EPR pair replaced with the GHZ state in \eqref{eq:ecmax} and \eqref{eq:edmin} (or any other generator, if the normalization is changed accordingly). It is easy to see that $\loccqrate{\rho}{\GHZ}$ is not additive and not convex, e.g. the tripartie pure states $\EPR_{AB}\otimes\ket{000}$ and $\EPR_{BC}\otimes\ket{111}$ cannot be transformed into GHZ states, but any nontrivial convex combination of them as well as their product is distillable.

As mentioned in the introduction, the characterization in Theorem~\ref{thm:loccqrateinf} is not entirely satisfactory because one of the defining properties of the functionals involves asymptotic transformations, which makes it difficult to decide if an entanglement measure belongs to $\functionals{k}$ or not. In practice, one shows instead monotonicity under single-shot LOCC transformations (which is also necessary) and asymptotic continuity, a particular type of continuity estimate depending logarithmically on the dimension of the Hilbert space. As this condition is only meaningful for quantities defined on all possible (finite) Hilbert space dimensions, it should be considered as a property of a function on
\begin{equation}\label{eq:allstatespaces}
\bigcup_{d_1,\ldots,d_k=1}^\infty\states(\complexes^{d_1}\otimes\cdots\otimes\complexes^{d_k}).
\end{equation}
For simplicity, we regard this set as a metric space in the following way: when $\rho,\sigma$ are states on the same Hilbert space then their distance is the purified distance, while if they live on different Hilbert spaces then we define their distance to be $1$ (as if they had orthogonal supports, although any positive constant would do). In the following $\dim\mathcal{H}$ will denote the function on the disjoint union \eqref{eq:allstatespaces} that takes the value $d_1d_2\cdots d_k$ on $\states(\complexes^{d_1}\otimes\cdots\otimes\complexes^{d_k})$.
\begin{definition}\label{def:asymptoticcontinuity}
A function $f:\bigcup_{d_1,\ldots,d_k=1}^\infty\states(\complexes^{d_1}\otimes\cdots\otimes\complexes^{d_k})\to\reals$ is \emph{asymptotically continuous} if
\begin{equation}
\frac{f}{1+\log\dim\mathcal{H}}
\end{equation}
is uniformly continuous.
\end{definition}
Any additive function satisfying this definition that is at the same time monotone under LOCC channels is also monotone under $\le$. The same definition can be used also if $f$ is defined on a subspace of \eqref{eq:allstatespaces}. In particular, we will study functions defined on pure states in Section~\ref{sec:pure}.
\begin{remark}
We note that there are slight variations in the literature on how asymptotic continuity is defined. Some authors require an estimate of the form $|f(\rho)-f(\sigma)|\le C_1\norm[1]{\rho-\sigma}\log\dim\mathcal{H}+C_2$ for some $C_1,C_2>0$ (as suggested by the Fannes inequality \cite{fannes1973continuity}), while others require $|E(\rho)-E(\sigma)|\le o(1)(1+\log\dim\mathcal{H})$ \cite{donald2002uniqueness} or $|E(\rho)-E(\sigma)|\le C\norm[1]{\rho-\sigma}\log\dim\mathcal{H}+o(1)$ \cite{synak2006asymptotic}, where $o(1)$ is any function that vanishes as $\norm[1]{\rho-\sigma}\to 0$. Our formulation is equivalent to the second one, underlining that asymptotic continuity is not a metric property, but depends only on the uniform structure. In particular, our condition does not change if we replace the purified distance with the trace norm distance, since they determine the same uniform structure \cite{fuchs1999cryptographic}. Nevertheless, as we show below, it \emph{does} imply an explicit continuity estimate on pure states.
\end{remark}

We stress that asymptotic continuity is only a sufficient condition for finding entanglement measures that are relevant in the asymptotic limit. It is possible for an entanglement measure to be not asymptotically continuous, but still provide an upper bound on certain rates, as the example of the logarithmic negativity shows \cite{vidal2002computable}. Theorem~\ref{thm:loccqrateinf} would become considerably stronger if one could show that every element of $\functionals{k}$ is asymptotically continuous in the sense of Definition~\ref{def:asymptoticcontinuity}.

\section{Entanglement measures on pure states}\label{sec:pure}

In this section we restrict our attention to pure states. The constructions from Section~\ref{sec:general} can also be applied to this case, resulting in a submonoid $M^{\textnormal{pure}}$ of $M$ consisting of (equivalence classes of) pure $k$-partite states. The set of normalized functionals $M^{\textnormal{pure}}\to\reals$ will be denoted by $\purefunctionals{k}$. As a notational simplification, we will write unit vectors $\varphi,\psi,\ldots$ as the argument of functionals on pure states with the understanding that $E(\varphi)\equiv E(\ketbra{\varphi}{\varphi})$.

Our aim is to prove equivalent characterizations of the elements of $\purefunctionals{k}$, emphasizing properties that can be verified without considering transformations in the asymptotic limit, i.e. involve only single-copy conditions. In the argument the direct sum operation plays a central role. We start with an inequality that is valid also for certain non-asymptotic measures.
\begin{proposition}\label{prop:superpositionlowerbound}
Let $E:M^\textnormal{pure}\to\reals$ be fully additive, monotone on average, and normalized to $E(\GHZ)=1$. Then for all $k$-partite state vectors $\varphi,\psi$ and $p\in[0,1]$ the inequality
\begin{equation}
E(\sqrt{p}\varphi\oplus\sqrt{1-p}\psi)\ge pE(\varphi)+(1-p)E(\psi)+h(p)
\end{equation}
holds.
\end{proposition}
\begin{proof}
Up to local unitary transformations, the $n$th tensor power of the direct sum can be written as
\begin{equation}
(\sqrt{p}\varphi\oplus\sqrt{1-p}\psi)^{\otimes n}=\bigoplus_{m=0}^n\sqrt{\binom{n}{m}p^m(1-p)^{n-m}}\varphi^{\otimes m}\otimes\psi^{\otimes(n-m)}\otimes\GHZ_{\binom{n}{m}}.
\end{equation}
The direct sum over $m$ determines a decomposition of the local Hilbert spaces into $n+1$ pairwise orthogonal subspaces. If every party performs the corresponding measurement, then the results will always be identical. With probability $\binom{n}{m}p^m(1-p)^{n-m}$ they obtain the outcome $m$ and the resulting state is $\varphi^{\otimes m}\otimes\psi^{\otimes(n-m)}\otimes\GHZ_{\binom{n}{m}}$. Using that $E$ is fully additive and monotone on average, we obtain
\begin{equation}
\begin{split}
nE(\sqrt{p}\varphi\oplus\sqrt{1-p}\psi)
 & = E((\sqrt{p}\varphi\oplus\sqrt{1-p}\psi)^{\otimes n})  \\
 & \ge \sum_{m=0}^n\binom{n}{m}p^m(1-p)^{n-m}E\left(\varphi^{\otimes m}\otimes\psi^{\otimes(n-m)}\otimes\GHZ_{\binom{n}{m}}\right)  \\
 & = \sum_{m=0}^n\binom{n}{m}p^m(1-p)^{n-m}\left[mE(\varphi)+(n-m)E(\psi)+E(\GHZ_{\binom{n}{m}})\right]  \\
 & = npE(\varphi)+n(1-p)E(\psi)+\sum_{m=0}^n\binom{n}{m}p^m(1-p)^{n-m}\log\binom{n}{m},
\end{split}
\end{equation}
using that the expected value of a binomial distribution with parameters $n$, $p$ is $np$. Divide by $n$ and let $n\to\infty$:
\begin{equation}
\begin{split}
E(\sqrt{p}\varphi\oplus\sqrt{1-p}\psi)
 & \ge \lim_{n\to\infty}\left[pE(\varphi)+(1-p)E(\psi)+\frac{1}{n}\sum_{m=0}^n\binom{n}{m}p^m(1-p)^{n-m}\log\binom{n}{m}\right]  \\
 & = pE(\varphi)+(1-p)E(\psi)+h(p).
\end{split}
\end{equation}
The last equality follows from a standard argument in the method of types \cite{csiszar2011information}. More specifically, it can be proved using the estimates $nh(m/n)-2\log(n+1)\le\log\binom{n}{m}\le nh(m/n)$ and the law of large numbers.
\end{proof}

In the following theorem we list equivalent conditions for entanglement measures on pure states to be monotone under asymptotic LOCC transformations. Together with Theorem~\ref{thm:loccqrateinf}, it implies our main result, Theorem~\ref{thm:main}. In the theorem below, asymptotic continuity is understood in a similar way as in Definition~\ref{def:asymptoticcontinuity}, but with the function defined only on pure states.
\begin{theorem}
Let $E:M^\textnormal{pure}\to\reals$ be a function that is fully additive, monotone on average and normalized to $E(\GHZ)=1$. The following are equivalent:
\begin{enumerate}
\item\label{it:Easymptoticallycontinuous} $E$ is asymptotically continuous,
\item\label{it:Epurefunctional} $E\in\purefunctionals{k}$,
\item\label{it:Edirectsum} for all $k$-partite state vectors $\varphi,\psi$ and $p\in[0,1]$ the equality
\begin{equation}
E(\sqrt{p}\varphi\oplus\sqrt{1-p}\psi)=pE(\varphi)+(1-p)E(\psi)+h(p)
\end{equation}
holds,
\item\label{it:Econtinuityestimate} for all $\varphi,\psi\in\mathcal{H}$ the continuity estimate
\begin{equation}\label{eq:Econtinuityestimate}
\left|E(\varphi)-E(\psi)\right|\le a(\purifieddistance(\ketbra{\varphi}{\varphi},\ketbra{\psi}{\psi}))\log\dim\mathcal{H}+b(\purifieddistance(\ketbra{\varphi}{\varphi},\ketbra{\psi}{\psi}))
\end{equation}
holds with
\begin{align}
a(\delta) & = \frac{\left(1+\delta^{\frac{2}{k+1}}\right)^{k+1}-1+\delta^2}{1-\delta^2}  \\
b(\delta) & = \frac{\left(1+\delta^{\frac{2}{k+1}}\right)^{k+1}}{1-\delta^2}h\left(\left(1+\delta^{\frac{2}{k+1}}\right)^{-1}\right).
\end{align}
\end{enumerate}
\end{theorem}
\begin{proof}\leavevmode

\ref{it:Easymptoticallycontinuous}$\implies$\ref{it:Epurefunctional}:
The main difficulty in proving this implication is that $E$ is only assumed to be defined on pure states, while the preorder allows $\varphi^{\otimes n}$ to be transformed to a mixed state close to $\psi^{\otimes n}$. To overcome this, we need to argue that any such protocol can be modified in such a way that the output is pure conditioned on a classical label, and the resulting pure state is also close to $\psi^{\otimes n}$ with high probability. The details are as follows.

Let $\varphi\ge\psi$ with $\psi\in\mathcal{H}=\mathcal{H}_1\otimes\cdots\otimes\mathcal{H}_k$. This means that there is a sequence of states $\rho_n$ such that $\ketbra{\varphi}{\varphi}^{\otimes n}\loccto\rho_n$ for all $n$ and $\purifieddistance(\rho_n,\ketbra{\psi}{\psi}^{\otimes n})\to 0$ as $n\to\infty$. There is nothing to prove if $E(\psi)=0$, so we will assume $E(\psi)>0$.

Let $\delta\in(0,(1+\log\dim\mathcal{H})^{-1}E(\psi))$ and choose $\epsilon$ such that for any Hilbert space $\mathcal{K}$ and unit vectors $\omega,\tau\in\mathcal{K}$, $\purifieddistance(\ketbra{\omega}{\omega},\ketbra{\tau}{\tau})\le\epsilon$ implies $|E(\omega)-E(\tau)|\le\delta(1+\log\dim\mathcal{K})$.

Choose an $n\ge 1$ and $\rho$ such that $\ketbra{\varphi}{\varphi}^{\otimes n}\loccto\rho$ and $\purifieddistance(\rho,\ketbra{\psi}{\psi}^{\otimes n})\le\epsilon^2$. The initial state $\ketbra{\varphi}{\varphi}^{\otimes n}$ is pure, therefore it is possible to modify the LOCC protocol in such a way that it results in an ensemble of pure states whose average is $\rho$, and the classical label is available to each party at the end (i.e. the new protocol keeps track of the intermediate measurement results, and communicates them to all parties). This means
\begin{equation}
\ketbra{\varphi}{\varphi}^{\otimes n}\loccto\sum_{x\in\mathcal{X}}P(x)\ketbra{\varphi_x}{\varphi_x}\otimes\ketbra{x}{x},
\end{equation}
where $\mathcal{X}$ is some index set (which we can assume to be finite, by approximating the protocol with a finite-round one and measurements with finitely many outcomes if necessary), $P\in\distributions(\mathcal{X})$ and
\begin{equation}
\rho=\sum_{x\in\mathcal{X}}P(x)\ketbra{\varphi_x}{\varphi_x}.
\end{equation}

By the assumption on $\rho$, we have
\begin{equation}
\begin{split}
\epsilon^4
 & \ge \purifieddistance(\rho,\ketbra{\psi}{\psi}^{\otimes n})^2  \\
 & = 1-\fidelity(\rho,\ketbra{\psi}{\psi}^{\otimes n})^2  \\
 & = 1-\bra{\psi^{\otimes n}}\rho\ket{\psi^{\otimes n}}  \\
 & = \sum_{x\in\mathcal{X}}P(x)\left[1-\braket{\psi^{\otimes n}}{\varphi_x}\braket{\varphi_x}{\psi^{\otimes n}}\right]  \\
 & = \sum_{x\in\mathcal{X}}P(x)\purifieddistance(\ketbra{\varphi_x}{\varphi_x},\ketbra{\psi}{\psi}^{\otimes n})^2.
\end{split}
\end{equation}
Let $A=\setbuild{x\in\mathcal{X}}{\purifieddistance(\ketbra{\varphi_x}{\varphi_x},\ketbra{\psi}{\psi}^{\otimes n})^2\ge \epsilon^2}$. By the Markov inequality we have $P(A)\le\frac{\epsilon^4}{\epsilon^2}=\epsilon^2$. We use that $E$ is additive, monotone on average, and the choice of $\epsilon$:
\begin{equation}
\begin{split}
nE(\varphi)
 & = E(\varphi^{\otimes n})  \\
 & \ge \sum_{x\in\mathcal{X}}P(x)E(\varphi_x)  \\
 & \ge \sum_{x\in\mathcal{X}\setminus A}P(x)E(\varphi_x)  \\
 & \ge \sum_{x\in\mathcal{X}\setminus A}P(x)\left[E(\psi^{\otimes n})-\delta(1+n\log\dim\mathcal{H})\right]  \\
 & = (1-P(A))\left[nE(\psi)-\delta(1+n\log\dim\mathcal{H})\right]  \\
 & \ge (1-\epsilon^2)\left[nE(\psi)-\delta(1+n\log\dim\mathcal{H})\right].
\end{split}
\end{equation}
Divide by $n$ and let $n\to\infty$, then $\epsilon\to 0$ and finally $\delta\to 0$ to get $E(\varphi)\ge E(\psi)$.

\ref{it:Epurefunctional}$\implies$\ref{it:Edirectsum}:
The inequality $E(\sqrt{p}\varphi\oplus\sqrt{1-p}\psi)\ge pE(\varphi)+(1-p)E(\psi)+h(p)$ is true by Proposition~\ref{prop:superpositionlowerbound}, therefore we only need to show the reverse inequality. Let $\delta>0$ and $n\in\naturals$, and consider the state
\begin{equation}
\omega=\varphi^{\otimes\lceil n(p+\delta)\rceil}\otimes\psi^{\otimes\lceil n(1-p+\delta)\rceil}\otimes(\sqrt{p}\ket{0\ldots 0}+\sqrt{1-p}\ket{1\ldots 1})^{\otimes n}.
\end{equation}
$t$ copies of this state can be written as a direct sum
\begin{equation}
\begin{split}
\omega^{\otimes t}
 & = \varphi^{\otimes t\lceil n(p+\delta)\rceil}\otimes\psi^{\otimes t\lceil n(1-p+\delta)\rceil}\otimes(\sqrt{p}\ket{0\ldots 0}+\sqrt{1-p}\ket{1\ldots 1})^{\otimes tn}  \\
 & = \bigoplus_{m=0}^{tn}\sqrt{\binom{tn}{m}p^m(1-p)^{tn-m}}\varphi^{\otimes t\lceil n(p+\delta)\rceil}\otimes\psi^{\otimes t\lceil n(1-p+\delta)\rceil}\otimes\GHZ_{\binom{tn}{m}}.
\end{split}
\end{equation}
The relation
\begin{equation}
\varphi^{\otimes t\lceil n(p+\delta)\rceil}\otimes\psi^{\otimes t\lceil n(1-p+\delta)\rceil} \loccto \varphi^{\otimes m}\otimes\psi^{\otimes tn-m}
\end{equation}
holds as long as $tn-t\lceil n(1-p+\delta)\rceil\le m\le t\lceil n(p+\delta)\rceil$, since the number of copies can be reduced by LOCC (tracing out local subsystems). For the remaining terms we use that the left hand side can be transformed into a separable state $\chi_m$. By \cite[Proposition 2.]{jensen2019asymptotic}, the transformations can be applied termwise in the direct sum, i.e.
\begin{multline}\label{eq:approximatesumpower}
\omega^{\otimes t}\loccto\bigoplus_{m=tn-t\lceil n(1-p+\delta)\rceil}^{t\lceil n(p+\delta)\rceil}\sqrt{\binom{tn}{m}p^m(1-p)^{tn-m}}\varphi^{\otimes m}\otimes\psi^{\otimes tn-m}\otimes\GHZ_{\binom{tn}{m}}  \\  \oplus\bigoplus_{\substack{m<tn-t\lceil n(1-p+\delta)\rceil  \\  \text{or }m>t\lceil n(p+\delta)\rceil}}\sqrt{\binom{tn}{m}p^m(1-p)^{tn-m}}\chi_m\otimes\GHZ_{\binom{tn}{m}}.
\end{multline}
On the other hand, $(\sqrt{p}\varphi\oplus\sqrt{1-p}\psi)^{\otimes tn}$ may be written as
\begin{equation}\label{eq:sumpower}
\bigoplus_{m=0}^{tn}\sqrt{\binom{tn}{m}p^m(1-p)^{tn-m}}\varphi^{\otimes m}\otimes\psi^{\otimes tn-m}\otimes\GHZ_{\binom{tn}{m}}.
\end{equation}
We choose the separable states $\chi_m$ to have nonnegative inner product with the corresponding term in \eqref{eq:approximatesumpower}, so that the overlap between the right hand side of \eqref{eq:approximatesumpower} and \eqref{eq:sumpower} is at least
\begin{equation}
\sum_{m=tn-t\lceil n(1-p+\delta)\rceil}^{t\lceil n(p+\delta)\rceil}\binom{tn}{m}p^m(1-p)^{tn-m},
\end{equation}
as can be seen by considering only the common terms in the direct sums. The limit of this sum as $t\to\infty$ is $1$, therefore $\omega\ge(\sqrt{p}\varphi\oplus\sqrt{1-p}\psi)^{\otimes n}$.

$E$ is assumed to be in $\purefunctionals{k}$, therefore
\begin{equation}
\begin{split}
nE(\sqrt{p}\varphi\oplus\sqrt{1-p}\psi)
 & = E((\sqrt{p}\varphi\oplus\sqrt{1-p}\psi)^{\otimes n})  \\
 & \le E(\omega)  \\
 & = \lceil n(p+\delta)\rceil E(\varphi)+\lceil n(1-p+\delta)\rceil E(\psi)+nh(p).
\end{split}
\end{equation}
We divide by $n$, let $n\to\infty$ and then $\delta\to 0$ to get $E(\sqrt{p}\varphi\oplus\sqrt{1-p}\psi)\le pE(\varphi)+(1-p)E(\psi)+h(p)$.

\ref{it:Edirectsum}$\implies$\ref{it:Econtinuityestimate}:
Let $\varphi,\psi\in\mathcal{H}=\mathcal{H}_1\otimes\cdots\otimes\mathcal{H}_k$ be unit vectors. We may assume that $\purifieddistance(\ketbra{\varphi}{\varphi},\ketbra{\psi}{\psi})\in(0,1)$, since there is nothing to prove otherwise. We consider the following construction with $A,B\in\complexes\setminus\{0\}$, $q,\lambda\in(0,1)$ to be chosen later, subject to the constraints
\begin{align}
1 & = |A|^2+|B|^2+2\Re\overline{A}B\braket{\varphi}{\psi}  \label{eq:omegaunit}  \\
A & = -\sqrt{\frac{q}{1-q}\left(\frac{\lambda}{1-\lambda}\right)^k}  \label{eq:Acancel}
.
\end{align}
\eqref{eq:omegaunit} ensures that $\omega:=A\varphi+B\psi$ is a unit vector. We will think of the direct sum $\sqrt{q}\varphi\oplus\sqrt{1-q}\omega$ as the vector
\begin{equation}
\sqrt{q}\ket{\varphi}\otimes\ket{0\ldots 0}+\sqrt{1-q}\ket{\omega}\otimes\ket{1\ldots 1}\in\mathcal{H}\otimes(\complexes^2)^{\otimes k}.
\end{equation}
Let each party perform a projective measurement on their local qubit in the basis $\sqrt{\lambda}\ket{0}+\sqrt{1-\lambda}\ket{1}$ and $\sqrt{1-\lambda}\ket{0}-\sqrt{\lambda}\ket{1}$. Of the $2^k$ possible combinations of the outcomes, we focus on the one where every party projects onto $\sqrt{\lambda}\ket{0}+\sqrt{1-\lambda}\ket{1}$. The result of the projection is
\begin{multline}
\left[I_\mathcal{H}\otimes\left(\sqrt{\lambda}\bra{0}+\sqrt{1-\lambda}\bra{1}\right)^{\otimes k}\right]\left(\sqrt{q}\ket{\varphi}\otimes\ket{0\ldots 0}+\sqrt{1-q}\ket{\omega}\otimes\ket{1\ldots 1}\right)  \\
\begin{split}
 & = \sqrt{q\lambda^k}\ket{\varphi}+\sqrt{(1-q)(1-\lambda)^k}\ket{\omega}  \\
 & = \left(\sqrt{q\lambda^k}+A\sqrt{(1-q)(1-\lambda)^k}\right)\ket{\varphi}+B\sqrt{(1-q)(1-\lambda)^k}\ket{\psi}  \\
 & = B\sqrt{(1-q)(1-\lambda)^k}\ket{\psi},
\end{split}
\end{multline}
where the last equality uses \eqref{eq:Acancel}. Let $u=\min\{q,|B|^2(1-q)(1-\lambda)^k\}$. $E$ is monotone on average, therefore
\begin{equation}\label{eq:sumprojectionmonotone}
\begin{split}
u E(\psi)
 & \le |B|^2(1-q)(1-\lambda)^k E(\psi)  \\
 & \le E(\sqrt{q}\ket{\varphi}\otimes\ket{0\ldots 0}+\sqrt{1-q}\ket{\omega}\otimes\ket{1\ldots 1})  \\
 & = q E(\varphi)+(1-q)E(\omega)+h(q).
\end{split}
\end{equation}
We rearrange and use that $E(\varphi)$ and $E(\omega)$ are at most $\log\dim\mathcal{H}$:
\begin{equation}
\begin{split}
u(E(\psi)-E(\varphi))
 & \le (q-u)E(\varphi)+(1-q)E(\omega)+h(q)  \\
 & \le (1-u)\log\dim\mathcal{H}+h(q),
\end{split}
\end{equation}
and divide by $u$:
\begin{equation}\label{eq:onesidedestimate}
E(\psi)-E(\varphi)\le\frac{1-u}{u}\log\dim\mathcal{H}+\frac{h(q)}{u}.
\end{equation}

To get a continuity estimate, we need that $u\to 1$ as $|\braket{\varphi}{\psi}|\to 1$. To ensure this, we choose the values of the parameters (non-optimally) as
\begin{align}
A & = -\frac{1}{\sqrt{1-|\braket{\varphi}{\psi}|^2}}  \\
B & = \frac{\overline{\braket{\varphi}{\psi}}}{\sqrt{1-|\braket{\varphi}{\psi}|^2}}  \\
q = \lambda & = \frac{1}{1+\sqrt[k+1]{1-|\braket{\varphi}{\psi}|^2}} = \frac{1}{1+\purifieddistance(\ketbra{\varphi}{\varphi},\ketbra{\psi}{\psi})^\frac{2}{k+1}}
.
\end{align}
Then the conditions \eqref{eq:omegaunit} and \eqref{eq:Acancel} are satisfied, and with the abbreviation $F=|\braket{\varphi}{\psi}|$ we have
\begin{equation}
\begin{split}
u
 & = |B|^2(1-q)(1-\lambda)^k  \\
 & = \frac{F^2}{1-F^2}\frac{1-F^2}{(1+\sqrt[k+1]{1-F^2})^{k+1}}  \\
 & = \frac{F^2}{(1+\sqrt[k+1]{1-F^2})^{k+1}}  \\
 & = \frac{1-\purifieddistance(\ketbra{\varphi}{\varphi},\ketbra{\psi}{\psi})^2}{\left(1+\purifieddistance(\ketbra{\varphi}{\varphi},\ketbra{\psi}{\psi})^\frac{2}{k+1}\right)^{k+1}}
.
\end{split}
\end{equation}

Since $u$ and $q$ only depend on the fidelity between $\varphi$ and $\psi$, the upper bound \eqref{eq:onesidedestimate} also holds for the absolute value of the left hand side. This proves \ref{it:Econtinuityestimate}.

\ref{it:Econtinuityestimate}$\implies$\ref{it:Easymptoticallycontinuous}: From the continuity estimate for $\varphi,\psi\in\mathcal{H}$ we get the inequality
\begin{equation}
\left|\frac{E(\varphi)}{1+\log\dim\mathcal{H}}-\frac{E(\psi)}{1+\log\dim\mathcal{H}}\right|
 \le a(\purifieddistance(\ketbra{\varphi}{\varphi},\ketbra{\psi}{\psi}))+b(\purifieddistance(\ketbra{\varphi}{\varphi},\ketbra{\psi}{\psi})),
\end{equation}
which only depends on the distance between $\varphi$ and $\psi$. Therefore $\frac{E}{1+\log\dim\mathcal{H}}$ is uniformly continuous.
\end{proof}

\begin{remark}
It follows from the implication \ref{it:Epurefunctional}$\implies$\ref{it:Econtinuityestimate} that on any $k$-partite pure state space $\purefunctionals{k}$ is a uniformly equicontinuous set of functions. In particular,
\begin{align}
\loccqrate{\ketbra{\varphi}{\varphi}}{\GHZ} & = \min_{E\in\purefunctionals{k}}E(\varphi)  \\
\intertext{and}
\frac{1}{\loccqrate{\GHZ}{\ketbra{\varphi}{\varphi}}} & = \max_{E\in\purefunctionals{k}}E(\varphi)
\end{align}
are also uniformly continuous, satisfying the same continuity estimate as in \ref{it:Econtinuityestimate}.
\end{remark}

We note that the restriction of an element of $\functionals{k}$ to $M^{\textnormal{pure}}$ is in $\purefunctionals{k}$, therefore is asymptotically continuous on pure states.

\section{Concluding remarks}

In this paper we proved characterizations of entanglement transformation rates in terms of sets of multipartite entanglement measures defined implicitly through axioms that they satisfy. Both in the general mixed state case and restricted to pure states, these axioms include full additivity, and in the case of pure states we showed that the required monotonicity under asymptotic LOCC transformations can be replaced with properties that do not involve an asymptotic limit: monotonicity on average and an explicit continuity estimate. These results suggest several open problems that we believe are worth investigating. The most obvious one is to find all the elements of $\functionals{k}$ or $\purefunctionals{k}$, which is without doubt a tremendous challenge. We describe some other open questions that may be less difficult to resolve.
\begin{itemize}
\item In connection with the possibility of reversible asymptotic entanglement transformations, it was shown in \cite{linden2005reversibility} that there are no reversible transformations between tripartite GHZ states and any combination of EPR pairs between pairs of parties. By Theorem~\ref{thm:main}, this can in principle be reproved using suitable elements of $\purefunctionals{3}$. The marginal entropies are in $\purefunctionals{3}$, and show that the only possibility would be that $\EPR_{AB}\otimes\EPR_{BC}\otimes\EPR_{AC}$ is asymptotically equivalent to $\GHZ^{\otimes 2}$. The task is therefore to find an element $E$ of $\purefunctionals{3}$ such that $E(\EPR_{AB})+E(\EPR_{BC})+E(\EPR_{AC})\neq 2$.
\item Is every element of $\functionals{k}$ asymptotically continuous? Our proof method for the pure case does not seem to have an analogue for mixed states. Since these functionals are convex, a possible route for proving the mixed case would be to show that $\rho\mapsto E(\rho)+\entropy(\rho)$ is concave, as in the case of the relative entropy of entanglement \cite{linden2005reversibility}, and argue as in \cite[Proposition 3.23]{christandl2006structure}.
\item The transformations that we consider allow a sublinear amount of quantum communication (or GHZ states) in the limit of many copies. It appears to be an open question if this actually helps or the rate would be the same without any quantum communication (see partial results in \cite{thapliyal2003multipartite}).
\item We have not made any attempt to optimize the continuity estimate \eqref{eq:Econtinuityestimate}, as it seems unlikely that the proof method leads to a significantly better bound. It would be interesting to see if the $\delta$-dependence in the first term can be improved from $O(\delta^{\frac{2}{k+1}})$ to $O(\delta)$ in general. By the Fannes inequality this is possible for the entropy of entanglement across any bipartite cut.
\item Finally, we sketch a possible route to characterizing the rates $\loccqrate{\ketbra{\varphi}{\varphi}}{\sigma}$ when the initial state is pure and distillable (globally entangled), in terms of the functionals defined on pure states (of which the present work offers a better understanding). To this end one could consider a formation-type extension of the rates between pure states:
\begin{equation}
\begin{split}
E_{F,\varphi}(\sigma)
 & = \inf_{\substack{(p_x,\psi_x)_{x\in\mathcal{X}}  \\  \sum_{x\in\mathcal{X}}p_x\ketbra{\psi_x}{\psi_x}=\rho}}\sum_{x\in\mathcal{X}}\frac{p_x}{\loccqrate{\ketbra{\varphi}{\varphi}}{\ketbra{\psi_x}{\psi_x}}}  \\
 & = \inf_{\substack{(p_x,\psi_x)_{x\in\mathcal{X}}  \\  \sum_{x\in\mathcal{X}}p_x\ketbra{\psi_x}{\psi_x}=\rho}}\sum_{x\in\mathcal{X}}p_x\sup_{E\in\purefunctionals{k}}\frac{E(\psi_x)}{E(\phi)}
,
\end{split}
\end{equation}
i.e. the convex roof extension of $\loccqrate{\ketbra{\varphi}{\varphi}}{\cdot}^{-1}$,
which should be an upper bound on $\loccqrate{\ketbra{\varphi}{\varphi}}{\sigma}^{-1}$. We expect that the regularization of $E_{F,\varphi}$ is equal to this rate, as is the case with the usual entanglement of formation and entanglement cost \cite{hayden2001asymptotic}. However, $E_{F,\varphi}(\sigma)$ is probably not additive for any $\varphi$, as it is known to be non-additive for $\varphi=\EPR$ \cite{shor2004equivalence,hastings2009superadditivity}.
\end{itemize}

\section*{Acknowledgement}

I thank Asger Kj{\ae}rulff Jensen for discussions. This work was supported by the \'UNKP-20-5 New National Excellence Program of the Ministry for Innovation and Technology and the Bolyai J\'anos Research Fellowship of the Hungarian Academy of Sciences. We acknowledge support from the Hungarian National Research, Development and Innovation Office (NKFIH) within the Quantum Technology National Excellence Program (Project Nr.~2017-1.2.1-NKP-2017-00001) and via the research grants K124152, KH129601.

\bibliography{refs}{}

\end{document}